\documentclass[letterpaper, 10 pt, conference]{ieeeconf}  
\IEEEoverridecommandlockouts

\usepackage{hyperref}  % make clickable links in LaTeX :-)

\usepackage{amsmath,amssymb,amsfonts}%
\usepackage{theorem}%
\usepackage{graphicx,color}%
\usepackage{subcaption}
\usepackage{epstopdf}
\usepackage{psfrag}

\newtheorem{theorem}{Theorem}

\newtheorem{lemma}[theorem]{Lemma}
\newtheorem{assumption}{Assumption}
\newtheorem{proposition}{Proposition}
\newtheorem{remark}{Remark}

\def\field#1{\mathbb #1}%
\def\R{\field{R}}%

\newcommand{\diag}{\mathrm{diag}}%

\def\K{\mathcal{K}}%

\allowdisplaybreaks%

\newcommand \sgn   {\text{sgn}}

\newcommand{\normt}[1]{{\left\vert\kern-0.25ex\left\vert\kern-0.25ex\left\vert #1 
    \right\vert\kern-0.25ex\right\vert\kern-0.25ex\right\vert}}
    
\title{Fixed-Time Convergent Control Barrier Functions for Coupled Multi-Agent Systems Under STL Tasks}%

\author{Maryam Sharifi and
	Dimos~V.~Dimarogonas \thanks{This work was supported by the ERC CoG LEAFHOUND, the Swedish Research
Council (VR) and the Knut \& Alice Wallenberg Foundation (KAW).
	\newline
	The authors are with the Department of Automatic Control, School of Electrical Engineering and Computer Science, KTH Royal Institute of Technology, Stockholm, Sweden. \texttt{\{msharifi, dimos\}@kth.se }.}
}
\date{}%

\begin{document}

\maketitle%

\begin{abstract}
This paper presents a control strategy based on a new notion of time-varying fixed-time convergent control barrier functions (TFCBFs) for a class of coupled multi-agent systems under signal temporal logic (STL) tasks. In this framework, each agent is assigned a local STL task regradless of the tasks of other agents. Each task may be dependent on the behavior of other agents which may cause conflicts on the satisfaction of all tasks. Our approach finds a robust solution to guarantee the fixed-time satisfaction of STL tasks in a least violating way and independent of the agents' initial condition in the presence of undesired violation effects of the neighbor agents. Particularly, the robust performance of the task satisfactions can be adjusted in a user-specified way.
\end{abstract}


\section{Introduction}
Recent technological advances in distributed sensing, computation and data management have enabled us to develop smart systems using collaborative multi-agent systems. These emergent applications are required to perform more complex task specifications which are typically formulated by temporal logics \cite{kloetzer2009automatic}. Among those, signal temporal logic (STL) is more beneficial as it is interpreted over continuous-time signals \cite{maler2004monitoring}, allows for imposing tasks with strict deadlines and introduces quantitative semantics known as robustness to the physical systems \cite{fainekos2009robustness}.

Control barrier functions \cite{ames2016control} guarantee the existence of a control law that renders a desired set forward invariant. The notions of input-to-state safety and robustness have appeared in \cite{kolathaya2018input} and \cite{xu2015robustness}. Nonsmooth, Higher order and time-varying control barrier functions are provided in \cite{glotfelter2017nonsmooth}, \cite{xiao2019control}  and \cite{xu2018constrained}, respectively.
Control Lyapunov functions are control design tools to obtain a number of specific performance criteria, such as, optimality,  transient behavior or robustness. In most of the modern emergent applications such as cyber physical systems, connected automated vehicles and networked control systems, the safety property of the system performance has become a part of control design \cite{romdlony2014uniting}.

We aim to consider a class of control-affine nonlinear coupled multi-agent systems under dependent spatiotemporal constraints. Under spatial constraints, the system trajectories should evolve in some \emph{safe} sets at all times, while visiting some \emph{goal} sets in specific time intervals. These kinds of constraints are common in safety-critical applications. In addition, temporal constraints pertain to the system convergence or a task completion within a fixed-time interval, and appear in time-critical applications.

In \cite{garg2020distributed}, a distributed control strategy for safety and fixed-time stability of multi-agent systems has been provided, while \cite{black2020quadratic} considers the problem for a single-agent system subject to disturbances. However, they assume that there are no dynamical couplings among agents and their initial conditions are inside the safe sets, and provide independent constraints for safety preservation and performance satisfaction, which may cause failures in the satisfiability of all specifications. Moreover, they use time-invariant control barrier functions which contain a lower degree of freedom in comparison to the time-varying ones, and may lead to inability in achieving more complex tasks.
We introduce a time-varying fixed-time convergent control barrier function notion to guarantee the satisfaction of a set of STL tasks by maintaining the safety as well as convergence to the specified safe sets within a finite-time interval, independent of the initial conditions of the system. 
%In particular, we design control barrier functions and control Lyapunov functions in a unified manner in order to not lose the appropriate performance.

We study multi-agent systems working under \emph{local} and possibly \emph{conflicting} specifications from a fragment of STL tasks. Each agent is subject to its local task, while the task itself may depend on the behavior of other agents. Therefore, all local tasks may possibly not be satisfiable at the same time. A robust fixed-time framework is presented to find a least violating solution using the notion of fixed-time stability in a more suitable way compared to the approach presented in \cite{lindemann2019control}. Particularly in this paper, the lower bound of the presented fixed-time convergent barrier function is tunable with respect to parameters of the quadratic programming formulation, independent of initial conditions, and the time of reaching this optimal bound is characterized in a user-specified way. 
Regarding the fixed-time stability properties we ensure that if the required conditions are not satisfied initially, they will be satisfied within a fixed-time and remain satisfied thereafter. 
%More importantly, when the desired behavior is satisfied at some point, it will remain satisfied thereafter. 
Therefore, we are able to unify the safety and performance criteria in one fixed-time constraint.
%We aim to satisfy the spatiotemporal constraints in a sequential way and in an event-triggered fashion. This way, a fixed preparation time ahead is determined to activate the corresponding fixed-time convergent barrier function of the next goal set.

Section \ref{setup} gives some preliminaries on STL, multi-agent systems and problem formulation. Problem solution is stated in Section \ref{solution} and simulations along with some concluding points are presented in Sections \ref{sim} and \ref{conc}, respectively.
\section{Preliminaries and problem formulation}\label{setup}
\subsection{Signal temporal logic (STL)}
Signal temporal logic (STL) \cite{maler2004monitoring} is based on predicates $\nu$ which are obtained by evaluation of a continuously differential predicate function $h: \mathbb{R}^{d}\to\mathbb{R}$ as $\nu:=\top$ (True) if $h(\xi)\geq 0$ and $\nu:=\bot$ (False) if $h(\xi)< 0$ for $\xi\in\mathbb{R}^{d}$. The STL syntax is then given by
\begin{align}\label{sintax}
\phi ::=\top |\nu|\neg\phi|\phi' \wedge {\phi ''}| \phi' U_{\left[ {a,b} \right]}{\phi ''},
\end{align}
where $\phi'$ and $\phi ''$ are STL formulas and where $U_{\left[ {a,b} \right]}$ is the until operator with $a\leq b<\infty$. In addition, we introduce $F_{\left[ {a,b} \right]}\phi:=\top U_{\left[ {a,b} \right]}\phi$ (eventually operator) and $G_{\left[ {a,b} \right]}\phi:=\neg F_{\left[ {a,b} \right]}\neg\phi$ (always operator). Let $\xi'\models\phi$ denote the satisfaction relation, i.e., whether a signal $\xi':\mathbb{R}_{\geq 0}\to\mathbb{R}^d$ satisfies $\phi$ (at time $0$). STL semantics are defined in \cite{maler2004monitoring}. A formula $\phi$ is satisfiable if $\exists\xi':\mathbb{R}_{\geq 0}\to\mathbb{R}^d$ such that $\xi'\models\phi$.
\subsection{Coupled multi-agent systems}
Consider an undirected graph $\mathcal{G}:=(\mathcal{V}, \mathcal{E})$ where $\mathcal{V}:=\{{1,\cdots,M}\}$ indicates the set consisting of $M$ agents and $\mathcal{E}\in \mathcal{V}\times\mathcal{V}$ represents communication links. Consider $x_k\in\mathbb{R}^{n_k}$ and $u_k\in\mathbb{R}^{m_k}$ as the state and input vectors of agent $k$, respectively. Furthermore, $x:=\left[{x_1^T,\cdots, x_M^T}\right]^T\in\mathbb{R}^{n}$ with $n:=n_1+\cdots+n_M$ and 
\begin{align}\label{agent_dyn}
\dot x_k=f_k(x_k,t)+g_k(x_k,t)u_k+c_k(x,t),
\end{align}
where $f_k:\mathbb{R}^{n_k}\times\mathbb{R}_{\geq 0}\to\mathbb{R}^{n_k}$, $g_k:\mathbb{R}^{n_k}\times\mathbb{R}_{\geq 0}\to\mathbb{R}^{n_k\times m_k}$ are locally Lipschitz continuous functions. In addition, $c_k(x,t)$ models dynamical couplings between agents such as mechanical connections, unmodelled dynamics or process noise. We assume that $c_k(x,t)$ is unknown but bounded. Therefore, the control design does not require any knowledge on $x$. In other words, there exist $C_k\geq 0$, which is known by agent $k$ and $\left\| {c_k(x,t)} \right\|\leq C_k$ for all $(x,t)\in\mathbb{R}^n\times\mathbb{R}_{\geq 0}$.

Each agent $k$ is assigned its local task $\phi_k$ of the form \eqref{sintax}. The satisfaction of $\phi_k$ may depend on the behavior of other agents $j \ne  k$, which is resulted by the evolution of their state trajectories. Therefore, the agent $k$  may obtain information from the other agent's tasks. 
We assume satisfaction of all local tasks is possible regardless of the other agent tasks.
However, since the tasks are dependent, satisfiability of each local task does not imply satisfiability of the conjunction of all local tasks. 
Let the satisfaction of $\phi_k$ depend on the behavior of a subset of agents denoted by ${\cal{V}}_k\subseteq\cal{V}$ with $\left| {{\cal{V}}_k} \right|\geq 1$ where $\left| {{\cal{V}}_k} \right|$ corresponds to the cardinality of the set ${\cal{V}}_k$. 
Let $\bar x_k:= \left[ {{x_j}_1^T \cdots {x_j}_{\left| {{\cal{V}}_k} \right|}^T} \right]^T$ be the stacked state vector of all agents in ${\cal{V}}_k$ for $j_1,\cdots,j_{\left| {{\cal{V}}_k} \right|}\in{\cal{V}}_k$ and $\bar n_k:={n_j}_1+\cdots+{n_j}_{\left| {{\cal{V}}_k} \right|}$. 
	We also define the projection map $p_k:\mathbb{R}^n\to\mathbb{R}^{\bar n_k}$ considering the fact that elements of $\bar x_k$ are contained in $x$. Let the projector from a set $\mathcal{S}\in\mathbb{R}^n$ onto the formula state-space $\mathbb{R}^{\bar n_k}$ be $P_k(\mathcal{S}):=\{\bar x_k\in\mathbb{R}^{\bar n_k}|\exists x\in\mathcal{S},p_k(x):=\bar x_k\}$. 
\subsection{Time-varying fixed-time convergent barrier functions}\label{solution3}
			Let $\mathfrak{H}^k(\bar x_k,t):\mathbb{R}^{\bar n_k}\times\mathbb{R}_{\geq 0}\to\mathbb{R}$ be a continuously differentiable function.
			Similar to \cite{lindemann2019decentralized}, we introduce time-varying barrier functions $\mathfrak{H}^k(\bar x_k,t)$ to satisfy STL task $\phi_k$.
		If 
		\begin{align*}
		\mathfrak{C}_k(t):=\{\bar x_k\in\mathbb{R}^{\bar n_k}|\mathfrak{H}^k(\bar x_k,t)\geq 0\}
		\end{align*}  
		is forward invariant, then it holds that $\bar x_k\models\phi_k$.
	 Similar to \cite{lindemann2019control} the barrier functions are piecewise continuous in the second argument with discontinuities caused by switchings at instants $\{s_0^k:=0,s_1^k,s_2^k,...\}$.
		Note that the time-varying barrier functions could be constructed for the conjunctions in $\phi_k$ by using a smooth under-approximation of the min-operator. In particular, for a number of $p_k$ functions $\mathfrak{H}^k_j(\bar x_k,t)$, we have that $\mathop {\min }\limits_{j \in \{ 1, \cdots ,{p_k}\} } \mathfrak{H}^k_j({{\bar x}_k},t) \approx -\frac{1}{\eta_k}\rm{ln}(\sum\limits_{\it{j} = 1}^{\it{p_k}} {\exp ( - {\it{\eta_k\mathfrak{H}^k_{j}}}({\it{{\bar x}_k}},t))} )$ with $\eta_k>0$, which is proportionally related to the accuracy of this approximation.
		 
		In view of~\cite[Steps A, B, and C]{lindemann2019decentralized}, each corresponding barrier function to $\phi_k$ could be constructed as 
		\begin{align}\label{barrier}
		   \mathfrak{H}^k(\bar x_k,t):=-\frac{1}{\eta_k}\rm{ln}(\sum\limits_{\it{j} = 1}^{\it{p_k}} {\exp ( - {\it{\eta_k}{\mathfrak{H}^k_j}}({\it{{\bar x}_k}},t))}),
		\end{align}
		 where each ${\it{\mathfrak{H}^k_j}}({\it{{\bar x}_k}},t)$ corresponds to an always or eventually operator with a corresponding time interval $\left[a^k_j,b^k_j\right]$. The switching instants $b_j^k$ are times that the $j$th temporal operator is satisfied and its corresponding barrier function ${\it{\mathfrak{H}^k_j}}({\it{{\bar x}_k}},t)$ will be deactivated. This time-varying strategy helps reducing the conservatism in the presence of large numbers of conjunctions. 
%		These switching sequences define a hybrid time domain 
		%as in~\cite[Chapter 2.2]{goebel2009hybrid} 
%		and justifies adopting solutions over $\left[s_0^k,s_1^k\right]$.
		Due to the knowledge of $\left[a^k_j,b^k_j\right]$, the switching sequences are known in advance and at time $t\geq s_i^k$, the next switch occurs at $s_{i+1}^k:=\rm{argmin}\it{_{b_j^k\in\{b_1^k,...,b_{p_k}^k\}}\zeta(b_j^k,t)}$ where $\zeta(b_j^k,t):=\left\{ \begin{array}{l}
b_j^k - t,\;\;b_j^k - t > 0\\
\infty,\;\;\;\;\;\;\;\;\rm{otherwise}
\end{array} \right.$.
	In addition, for each switching instant $s_l^k$, it holds that $\mathop {\lim }\limits_{\tau  \to s_l^ k- } {\mathfrak{C}_k}(\tau ) \subseteq {\mathfrak{C}_k}({s_l^k})$ where $\mathop {\lim }\limits_{\tau  \to s_l^ k- } {\mathfrak{C}_k}(\tau )$ is the left-sided limit of ${\mathfrak{C}_k}({t})$ at $t=s_l^k$.
		%\blue{(to be in the same set right after the switch)}.

	We also make the following assumption:
	\begin{assumption}
		The functions $\mathfrak{H}^k(\bar x_k,t)$, $k\in\{{1,\cdots,K}\}$, are differentiable, the sets $\mathfrak{C}_k$ are compact, and their interior (i.e., ${\rm{int}}(\mathfrak{C}_k(t))=\{\bar x_k|\mathfrak{H}^k(\bar x_k,t)>0\})$ is non-empty for all $t\geq 0$.
	\end{assumption} 
\subsection{Problem formulation}
We consider the STL fragment
\begin{subequations}
	\begin{align}
&\psi::= \top|\nu|\psi'\wedge\psi'',\label{1st}\\
&\phi::=G_{\left[ {a,b} \right]}\psi|F_{\left[ {a,b} \right]}\psi|\psi'U_{\left[ {a,b} \right]}\psi''|\phi'\wedge\phi'',\label{2nd}
\end{align}
\end{subequations}
where $\psi',\psi''$ are formulas of class $\psi$ in \eqref{1st} and $\phi',\phi''$ are formulas of class $\phi$ in \eqref{2nd}. Consider $K$ formulas $\phi_1,\cdots,\phi_K$ of the form \eqref{2nd} and let the satisfaction of $\phi_k$ for $k\in\{{1,\cdots,K}\}$ depend on the set of agents $\mathcal{V}_k\subseteq\mathcal{V}$.
\begin{assumption}\label{concave}
	All predicate functions in $\phi_k$ are concave.
	\end{assumption}
	Concave predicate functions contain linear functions as well as functions corresponding to reachability tasks (predicates like $\left\| x-p \right\|\leq \epsilon$, $p\in\R^n$, $\epsilon\geq 0$). As the minimum of concave predicate functions is again concave, 
	%and local optima conditions don't happen this way
	 concave predicates are needed to construct valid control Lyapunov functions.
	
	Moreover, the formula dependencies should hold according to the graph topology as below.

\begin{assumption}\label{ass1}
	For each $\phi_k$ with $k\in\{{1,\cdots,K}\}$, 	it holds that $(j,k)\in\cal{E}$ for all $j\in{\cal{V}}_k\backslash \{k\}$. 
	\end{assumption}
%	Assume further that the sets of agents $\mathcal{V}_1,\cdots,\mathcal{V}_K\in\mathcal{V}$ are disjoint. i.e., $\mathcal{V}_{k_1}\cap \mathcal{V}_{k_2}=\emptyset$ for all $k_1,k_2\in\{{1,\cdots,K}\}$ with $k_1\ne k_2$, and $\mathcal{V}_{1} \cup\cdots\cup \mathcal{V}_{K}=\mathcal{V}$.
	
	 %Besides, $\mathfrak{h}^k({\bar{x}_k},t):\mathbb{R}^{\bar n_k}\times\mathbb{R}_{\geq 0}\to\mathbb{R}$ a continuously differentiable function defining the goal set $S^g_k(t)=\{\bar x_k|\mathfrak{h}^k({\bar{x}_k},t)\geq 0\}$ for each formula $\phi_k$. 
	% \blue{In other words, the function $\mathfrak{b}^k(\bar x_k,t)$ corresponds to the always operators and $\mathfrak{h}^k(\bar x_k,t)$ regards the eventually ones.}
	We further examine the behavior of each agent $k$ under satisfaction of the following assumption for other agents $j\ne k$, which we put in more perspective later (cf. Remark \ref{re1}).
	\begin{assumption}\label{ass2}
		Each agent $j\ne k$ applies a bounded and continuous control law $u_j(x,t)$ to achieve $x_j(t)\in\mathfrak{B}_j$ for a compact set $\mathfrak{B}_j$ and for all $t\geq 0$.
	\end{assumption}
	Considering \eqref{agent_dyn}, we can rewrite the stacked dynamics for the set of agents in $\mathcal{V}_k$ as follows
	\begin{align}\label{multi_sys}
	{\dot{\bar{x}}}_k=&\bar f_k(\bar x_k,t)+\bar g_k(\bar x_k,t)\bar u_k+\bar c_k(x,t)\nonumber\\&
	\tilde f_k(x_k,t)+\tilde g_k(x_k,t)u_k+\tilde c_k(x,t),
	\end{align}
	where $\bar f_k(\bar x_k,t):=\left[{f_j}_1({x_j}_1,t)^T,\cdots,{f_j}_{\left| {{\cal{V}}_k} \right|}({x_j}_{\left| {{\cal{V}}_k} \right|},t)^T\right]^T$,\\ $\bar g_k(\bar x_k,t):=\diag({g_j}_1({x_j}_1,t),\cdots,{g_j}_{\left| {{\cal{V}}_k} \right|}({x_j}_{\left| {{\cal{V}}_k} \right|},t))$,\\ $\bar c_k(x,t):=\left[{c_j}_1({x_j}_1,t)^T,\cdots,{c_j}_{\left| {{\cal{V}}_k} \right|}({x_j}_{\left| {{\cal{V}}_k} \right|},t)^T\right]^T$,
	and
	$\bar u_k:=\left[{u_j}_1^T,\cdots,{u_j}_{\left| {{\cal{V}}_k} \right|}^T\right]^T$
	for $j_1,\cdots,j_{\left| {{\cal{V}}_k} \right|}\in\mathcal{V}_k$.
	
	Therefore, $\tilde f_k(x_k,t):=\left[f_k(x_k,t)^T,0^T,\cdots,0^T\right]^T$, $\tilde g_k(x_k,t):=\left[g_k(x_k,t)^T,0^T,\cdots,0^T\right]^T$, $\tilde c_k(x,t):=\bar c_k(x,t)+\left[0^T,{d_j}_1(x,t)^T,\cdots,{d_j}_{\left| {{\cal{V}}_k} \right|}(x,t)^T\right]^T$ with ${d_j}(x,t):=f_j(x_j,t)+g_j(x_j,t)u_j(x,t)$.
	In the sequel, $\tilde c_k(x,t)$ is treated as an unknown disturbance. Let $\tilde C_k$ be a positive constant such that $\left\| {\tilde c_k(x,t)} \right\|\leq \tilde C_k$ for all $(x,t)\in\mathfrak{D}\times\mathbb{R}_{\geq 0}$ with $\mathfrak{D}\in\mathbb{R}^n$ an open and bounded set for which it holds that $P_k(\mathfrak{D}) \supset \mathfrak{C}_k(t)$ for all $t\geq 0$ as well as $P_j(\mathfrak{D}) \supset \mathfrak{B}_j(t)$ for all $j\ne k$. Due to Assumption \ref{ass2} and continuity property of functions $f_j(x_j,t)$ and $g_j(x_j,t)$, $\tilde C_k$ exists and acts as a non-vanishing disturbance. This will be elaborated more in Remark \ref{re1}.
	 \begin{assumption}\label{ass0}
	The function $g_k(x_k,t)$ has full row rank for $(x_k,t)\in\mathbb{R}^{n_k}\times\mathbb{R}_{\geq 0}$.
\end{assumption}
Assumption \ref{ass0} allows to decouple the construction of barrier functions from the agent dynamics. In other words, for a function $\mathfrak{H}^k({{\bar x}_k},t)$ it holds that $\frac{{\partial \mathfrak{H}^k({\bar{x}_k},t)}}{{\partial {x_k}}}{g_k}({x_k},t)=0$ if and only if $\frac{{\partial \mathfrak{H}^k({{\bar x}_k},t)}}{{\partial {x_k}}}=0$. This restriction could be relaxed for some class of dynamics using the notion of higher order barrier functions \cite{xiao2019control}.

	 	We should emphasize that if $\phi_k$ contains concave predicate functions and $\bar g_k(\bar x_k,t)$ has full row rank for all $(\bar x_k,t)\in\R^{\bar n_k}\times\R_{\geq 0}$, then $\mathfrak{H}^k(\bar x_k,t)$ can be constructed as in \cite{lindemann2019decentralized}.

	 The problem formulation is stated as follows:
	
	\textbf{Problem. 1} Find a control input $u_k(t)\in\mathcal{U}_k$, $t\geq 0$, $k\in\{{1,\cdots,K}\}$, such that for all initial conditions $\bar x_k(0)$ and in the absence of formulae dependencies and dynamic couplings, the set $\mathfrak{C}_k$ is invariant for \eqref{multi_sys}. In addition,  in the presence of such undesirable effects, the trajectories converge to a neighborhood of set $\mathfrak{C}_k$ in a fixed-time interval and independent of the initial condition of the agents; i.e., $\bar x_k(\bar T_k)\in \mathfrak{C}_k$ in a least violating way, for some user-defined $\bar T_k>0$.
		\section{Problem solution}\label{solution}
			In order to guarantee reaching the spatiotemporal constraints in the presence of non-vanishing additive disturbance in a least violating manner, we present \emph{fixed-time convergent control barrier functions} that are essential for valid behavior composition. 
			\subsection{Fixed-time convergence}
			We start with a lemma on the fixed-time convergence guarantee for a class of control Lyapunov functions (CLFs).
				\begin{lemma}\label{RFxT}
					\cite{black2020quadratic} A continuously differentiable positive-definite proper function $V_k:\mathbb{R}^{\bar n_k}\to\mathbb{R}_{\geq 0}$ is called robust fixed-time CLF (RFxT CLF) for \eqref{multi_sys}, if the following holds:
					\begin{align}\label{fx}
					 \dot V_k(\bar x_k)  \le  - {a_{1k}}V_k^{{b_{1k}}}(\bar x_k) - {a_{2k}}V_k^{{b_{2k}}}(\bar x_k) + {a_{3k}},
					\end{align}
					with $a_{1k},a_{2k}>0, a_{3k}\in\mathbb{R}, b_{1k}=1+\frac{1}{\mu_k}, b_{2k}=1-\frac{1}{\mu_k}$ for some $\mu_k>1$, along the trajectories of \eqref{multi_sys}. Then, there exists a neighborhood $D_k$ of the origin such that for all $\bar x_k(0)\in\mathbb{R}^{\bar n _k}\backslash D_k$, the trajectories of \eqref{multi_sys} reach the set $D_k$  within a fixed time $T_k$ satisfying 
					\begin{align}\label{T}
						T_k \le \left\{ \begin{array}{l}
							\frac{\mu_k }{{{a_{1k}}(c_k - b_k)}}\log (\frac{{\left|1+ c_k \right|}}{{\left|1+ b_k \right|}})\;\;\;\;\;;{a_{3k}} > 2\sqrt {{a_{1k}}{a_{2k}}} \\
							\frac{\mu_k }{{\sqrt {{a_{1k}}{a_{2k}}} }}(\frac{1}{{ k_k-1}})\;\;\;\;\;\;\;\;\;\;\;\;\;\;\;;{a_{3k}} = 2\sqrt {{a_{1k}}{a_{2k}}} \\
							\frac{\mu_k }{{{a_{1k}}{k_{1k}}}}(\frac{\pi }{2} - {\tan ^{ - 1}}{k_{2k}})\;\;\;\;;0\leq{a_{3k}} < 2\sqrt {{a_{1k}}{a_{2k}}}\\
							\frac{\mu_k \pi}{2{\sqrt {{a_{1k}}{a_{2k}}} }}\;\;\;\;\;\;\;\;\;\;\;\;\;\;\;\;\;\;\;\;\;\;\;\;; {a_{3k}}\leq 0
						\end{array} \right.,
					\end{align}
					with
					\begin{align}\label{D}
						D_k = \left\{ \begin{array}{l}
							\left\{ {\bar x_k|V_k \leq {{(\frac{{{a_{3k}} + \sqrt {{a_{3k}}^2 - 4{a_{1k}}{a_{2k}}} }}{{2{a_{1k}}}})}^{\mu_k} }} \right\}\\\;\;\;\;\;\;\;\;\;\;\;\;\;\;\;\;\;\;\;\;\;\;\;\;\;\;\;\;\;\;\;\;\;\;\;\;\;;{a_{3k}} > 2\sqrt {{a_{1k}}{a_{2k}}} \\
							\left\{ {\bar x_k|V_k \leq {k_k^{{\mu_k}} }{{(\frac{{{a_{2k}}}}{{{a_{1k}}}})}^{\frac{\mu_k }{2}}}} \right\};{a_{3k}} = 2\sqrt {{a_{1k}}{a_{2k}}} \\
								\left\{ {\bar x_k|V_k \leq {{\frac{{{a_{3k}} }}{{2{\sqrt {{a_{1k}}{a_{2k}}} }}}}}} \right\}\;\;\;\;;0\leq{a_{3k}} < 2\sqrt {{a_{1k}}{a_{2k}}}\\
							0^{\bar{n}_k}\;\;\;\;\;\;\;\;\;\;\;\;\;\;\;\;\;\;\;\;\;\;\;\;\;\;\;\;\;\;\;;{a_{3k}\leq 0}
						\end{array} \right.,
					\end{align}
					where $k_k>1$ and $b_k, c_k$ are the solutions of $\gamma_k(s)=a_{1k}s^2-a_{3k}s+a_{2k}=0$. Moreover, $k_{1k}=\sqrt{\frac{4a_{1k}a_{2k}-a_{3k}^2}{4a_{1k}^2}}$ and $k_{2k}=-\frac{a_{3k}}{\sqrt{{4a_{1k}a_{2k}-a_{3k}^2}}}$.
					%It is obvious that as long as $a_{3k}<2\sqrt {{a_{1k}}{a_{2k}}}$, every initial condition for the agent trajectories guarantee fixed-time convergence to the origin (goal set).
				\end{lemma}
				\begin{proof}
				For ${a_{3k}\leq 0}$, we obtain the standard form of the inequality which guarantees the fixed-time convergence to the origin for all $\bar x_k\in\mathbb{R}^{\bar n_k}$ (\cite{polyakov2011nonlinear}). For ${a_{3k}\geq 0}$, by rewriting \eqref{fx} we get
				\begin{align}\label{int}
I=&\int\limits_{{V_k}(\bar x_k(0))}^{{V_k}(\bar x_k({T_k}))}\frac{1}{{ - {a_{1k}}V_k^{{b_{1k}}} - {a_{2k}}V_k^{{b_{2k}}} + {a_{3k}}}}d{V_k}\nonumber\\&\ge \int\limits_0^{{T_k}} {dt = {T_k}},  
				\end{align}
				where $T_k$ is convergence time of the system trajectories to the set $D_k$. It can be shown that for all $\bar x_k\notin D_k$, the system trajectories reach the set $D_k$ in a fixed-time interval.
				
			To prove this claim, first consider $0\leq {a_{3k}}< {2{\sqrt {{a_{1k}}{a_{2k}}} }}$. We have that $- {a_{1k}}V_k^{{b_{1k}}} - {a_{2k}}V_k^{{b_{2k}}} + {a_{3k}}\leq -2{\sqrt {{a_{1k}}{a_{2k}}} }\bar V_k+{a_{3k}}$ for all $\bar V_k\geq {\frac{{{a_{3k}} }}{{2{\sqrt {{a_{1k}}{a_{2k}}} }}}}$. Thus, for all $V_k(\bar x_k(0))\geq \bar V_k\geq 1$ the left integrand in \eqref{int} is negative and hence, the following is obtained:
					\begin{align}\label{int2}
\int\limits_{{V_k}(\bar x_k(0))}^{1} \frac{d{V_k}}{{ - {a_{1k}}V_k^{{b_{1k}}} - {a_{2k}}V_k^{{b_{2k}}} + {a_{3k}}}} \le\nonumber\\ \int\limits_{{V_k}(\bar x_k(0))}^{1} \frac{d{V_k}}{{ - {a_{1k}}V_k^{{b_{1k}}} - {a_{2k}}V_k^{{b_{2k}}} + {a_{3k}}V_k}}.
				\end{align}
				We obtain $T_k\le I\le \frac{\mu_k }{{{a_{1k}}{k_{1k}}}}(\frac{\pi }{2} - {\tan ^{ - 1}}{k_{2k}})$ by evaluating the second integral in \eqref{int2}.
				
				For ${a_{3k}}\ge {2{\sqrt {{a_{1k}}{a_{2k}}} }}$ we have $\bar V_k\geq 1$. Therefore, for $V_k(\bar x_k)\geq \bar V_k\geq 1$ we get $- {a_{1k}}{V_k^{{b_{1k}}}} - {a_{2k}}{V_k^{{b_{2k}}}} + {a_{3k}}\leq - {a_{1k}}{V_k^{{b_{1k}}}} - {a_{2k}}{V_k^{{b_{2k}}}} + {a_{3k}}V_k$ which leads to 
					\begin{align*}
I \le \int\limits_{{V_k}(\bar x_k(0))}^{\bar V_k} \frac{d{V_k}}{{ - {a_{1k}}V_k^{{b_{1k}}} - {a_{2k}}V_k^{{b_{2k}}} + {a_{3k}}V_k}}.
				\end{align*}
				Solving the above integral for $\bar V_k\geq 1$ leads to $I\le \frac{\mu_k }{{{a_{1k}}(c_k - b_k)}}\log (\frac{{\left|1+ c_k \right|}}{{\left|1+ b_k \right|}})$ with $c_k\ge b_k$.
				
				Finally, for ${a_{3k}}={2{\sqrt {{a_{1k}}{a_{2k}}} }}$ we have $c_k = b_k=-\sqrt{\frac{a_{2k}}{a_{1k}}}$. Hence,
					\begin{align*}
I &\le \int\limits_{{V_k}(\bar x_k(0))}^{\bar V_k} \frac{d{V_k}}{{ - {a_{1k}}V_k^{{b_{1k}}} - {a_{2k}}V_k^{{b_{2k}}} + {a_{3k}}V_k}}\\&=\frac{\mu_k}{a_{1k}}(\frac{1}{c_k+\bar V^{\frac{1}{\mu_k}}}-\frac{1}{c_k+{V_k}(\bar x_k(0))})\\
&\le \frac{\mu_k}{a_{1k}}\frac{1}{c_k+\bar V^{\frac{1}{\mu_k}}}\le\frac{\mu_k}{\sqrt{a_{1k}a_{2k}}(k_k-1)},
				\end{align*}
				where the last inequality follows from the fact that $\bar V_k^{\frac{1}{\mu_k}}\ge -k_kc_k$ for $k_k>1$ results in a finite non-negative value for $I$.
				The proof is complete.
				\end{proof}
			\begin{remark}
				Note that an upper-bound for $T_k$ could be considered as a user-defined fixed convergence time.
				\end{remark}
					 Next, we provide a theorem to guarantee the robust fixed-time forward invariance property of the set $\mathfrak{C}_k(t)$. 
					 By the term \emph{robust} we mean that in the absence of agent couplings and violating effects of the other local tasks, the fixed-time convergence to the set  $\mathfrak{C}_k(t)$ is guaranteed. However, in the presence of such undesirable effects, the fixed-time convergence to the set $\mathfrak{C}_{k,rf}(t)\supseteq\mathfrak{C}_k(t)$, which later will be defined by Proposition \ref{convergence}, is guaranteed.
					% 		\begin{assumption}\label{ass3}
		%	For an extended locally Lipschitz continuous class $\mathcal{K}$ function $\alpha_k$ and for $(\bar x_l,t)\in P_k(\mathfrak{D})\times(s_j^k,s^{k}_{j+1})$ with $\frac{{\partial {\mathfrak{H}^k}({{\bar x}_k},t)}}{{\partial {\bar{x}_k}}}=0^T$ it holds that
		%	$\frac{{\partial {\mathfrak{H}^k}({{\bar x}_k},t)}}{{\partial {t}}}>  - {\alpha _k}\left( {{\mathfrak{H}^k}({{\bar x}_k},t)} \right)$.
		%	\end{assumption}
			
				\begin{theorem}
					Consider a multi-agent network consisting of $M$ agents subject to the dynamics of \eqref{agent_dyn} under Assumption \ref{ass0} and $K$ formulas $\phi_k$ of the form \eqref{2nd} under Assumption \ref{concave}.
					Let $\mathfrak{H}^k(\bar x_k,t)$ be a time-varying barrier function associated with the task $\phi_k$ according to Section \ref{solution3}.
				%	Let each $u_k(x_k,t)$ be locally bounded and measurable.
					If for some positive constants $\alpha_k$, $\beta_k$, $\gamma_{1k}>1$, $\gamma_{2k}<1$, for some open set $P_k(\mathfrak{D})$ with $P_k(\mathfrak{D})\supset \mathfrak{C}_k(t)$ for all $t\geq 0$, and for all $(\bar x_k,t)\in P_k(\mathfrak{D})\times(s^k_j,s^{k}_{j+1})$, there exists a  control law $u_k(x_k,t)$ for agent $k\in\mathcal{V}_k$ such that
					\begin{equation}\label{Ineq1}
					\begin{array}{l}
					\frac{{\partial 	\mathfrak{H}^k({{\bar x}_k},t)}}{{\partial {x_k}}}\left( {{f_k}({x_k},t) + {g_k}({x_k},t){u_k}} \right)+ \frac{{\partial 	\mathfrak{H}^k({{\bar x}_k},t)}}{{\partial t}}\ge\\\left\| {\frac{{\partial \mathfrak{H}^k({{\bar x}_k},t)}}{{\partial {\bar x}_k}}} \right\|\tilde C_k- {\alpha _k}\;\sgn({	\mathfrak{H}^k({{\bar x}_k},t)}){	|\mathfrak{H}^k({{\bar x}_k},t)|}^{\gamma _{1k}}\\ -{\beta _k}\;\sgn({	\mathfrak{H}^k({{\bar x}_k},t)}){	|\mathfrak{H}^k({{\bar x}_k},t)|}^{\gamma _{2k}},
					\end{array}
					\end{equation}
					then $\mathfrak{C}_k(t)$ is robust fixed-time forward invariant and $\mathfrak{H}^k({{\bar x}_k},t)$ is a valid time-varying fixed-time convergent control barrier function (TFCBF).
					%Hence, $\bar x_k\models\phi_k$.
				\end{theorem}
\begin{remark}\label{discuss}
We have substituted $\left\| {\frac{{\partial \mathfrak{H}^k({{\bar x}_k},t)}}{{\partial {\bar x}_k}}} \right\|\tilde C_k$ as an upper-bound for $\frac{{\partial \mathfrak{H}^k({{\bar x}_k},t)}}{{\partial {\bar x}_k}}\tilde c_k$ in the valid control barrier function condition \eqref{Ineq1}, since that way it may contain feasibility issues if $\frac{{\partial 	\mathfrak{H}^k({{\bar x}_k},t)}}{{\partial {x_k}}}{g_k}({x_k},t)=0$ and $\frac{{\partial \mathfrak{H}^k({{\bar x}_k},t)}}{{\partial {\bar x_k}}}\tilde c_k(x,t)\ne 0$. Then, satisfaction of the inequality would rely on $\frac{{\partial 	\mathfrak{H}^k({{\bar x}_k},t)}}{{\partial {\bar x_k}}}\tilde c_k(x,t)$ which comes from the behavior of the ${\cal{V}}_k\backslash \{k\}$ that are unknown to agent $k$. As mentioned before, we treat this term as an unknown disturbance and give an estimation for $\tilde C_k$ in the sequel.
Furthermore, consider $\rho_k$ as some extended class $\K$ function \cite{ames2016control} with
\begin{align}\label{rho}
 \rho_k(r)= \alpha_k\sgn(r)|r|^{\gamma_{1k}}+\beta_k\sgn(r)|r|^{\gamma_{2k}}.  
\end{align}
By Assumptions \ref{concave} and \ref{ass0},  the functions $\mathfrak{H}^k({{\bar x}_k},t)$  can be constructed with $\rho_k$
satisfying~\cite[Lemma 4]{lindemann2019decentralized} to ensure that $\frac{{\partial {\mathfrak{H}^k}({{\bar x}_k},t)}}{{\partial {t}}}>  - {\rho _k}\left( {{\mathfrak{H}^k}({{\bar x}_k},t)} \right)+ \chi$ for some $\chi>0$ when $\frac{{\partial {\mathfrak{H}^k}({{\bar x}_k},t)}}{{\partial {\bar{x}_k}}}\bar g_k(\bar x_k,t)=0$. 
%(Note that it is sufficient to consider $\rho_k(r)=\kappa_k r$ with $\kappa_k\leq 2\sqrt{\alpha_k\beta_k}$ to satisfy \eqref{rho}.) 
This ensures that all agents in $\mathcal{V}_k$ can use a collaborative control law as presented in~\cite[Theorem 1]{lindemann2019decentralized} and choosing \eqref{rho} gives the fixed-time convergence property without causing feasibility problems for \eqref{Ineq1}. Then, possible violation in \eqref{Ineq1} comes from conflicting local objectives.
We will treat the task conflictions by a relaxation term $\varepsilon_k$ in the quadratic program formulation (cf. Section \ref{QP_for}).
\end{remark}
	We defined a class of control Lyapunov functions (RFxT CLFs) with a user-defined fixed-time convergence guarantee in Lemma \ref{RFxT} with the convergence time (set) (i.e., $T_k$ ($D_k$)), characterized by given parameters $a_{1k}$, $a_{2k}$, $b_{1k}$, $b_{2k}$ and independent of the initial conditions $\bar x_k(0)$. The following Proposition proves that the inequality \eqref{Ineq1} leads to a robust fixed-time convergence to the predefined predicates.
	\begin{proposition}\label{convergence}
		Consider the set $\mathfrak{C}_k(t)$ associated with $\mathfrak{H}^k(\bar x_k,t)$ defined on $P_k(\mathfrak{D})$ with $ \mathfrak{C}_k(t)\subset P_k(\mathfrak{D})$. Let positive constants $\alpha_k$, $\beta_k$, $\delta_k$, $\gamma_{1k}=1+\frac{1}{\mu_k}$, $\gamma_{2k}=1-\frac{1}{\mu_k}$, $\mu_k>1$, be given. Then, any controller $u_k:P_k(\mathfrak{D}) \to\mathcal{U}_k$ such that \eqref{Ineq1} is satisfied for the system \eqref{multi_sys} with $\left\| {\tilde c_k(x,t)} \right\|\leq \tilde C_k$ and $\delta_k\geq\left\| {\frac{{\partial \mathfrak{H}^k({{\bar x}_k},t)}}{{\partial {\bar x}_k}}} \right\|\tilde C_k$, for all $(\bar x_k,t)\in P_k(\mathfrak{D})\times\mathbb{R}_{\geq 0}$, renders the set $\mathfrak{C}_k(t)$ robust fixed-time convergent. In particular, given the initial condition $\bar x_k(0)\in P_k(\mathfrak{D})\backslash \mathfrak{C}_k(0)$, the controller drives the state trajectories $\bar x_k(t)$ within a fixed-time given by \eqref{T} to the set $\mathfrak{C}_{k,rf}(t)$ given as follows. 
		\begin{align*}
		\mathfrak{C}_{k,{rf}}(t):=\{\bar x_k\in\mathbb{R}^{\bar n_k}|\mathfrak{H}^k(\bar x_k,t)\geq -\epsilon_{k,\max}\},
		\end{align*}  					
					where
					\begin{align*}
\epsilon_{k,\max}= \left\{ \begin{array}{l}
							  {(\frac{{{\delta_{k}} + \sqrt {{\delta_{k}}^2 - 4{\alpha_{k}}{\beta_{k}}} }}{{2{\alpha_{k}}}})}^{\mu_k}\;\;;{\delta_{k}} > 2\sqrt {{\alpha_{k}}{\beta_{k}}} \\
							{k_k^{{\mu_k}} }{{(\frac{{{\beta_{k}}}}{{{\alpha_{k}}}})}^{\frac{\mu_k }{2}}}\;\;\;\;\;\;\;\;\;\;\;\;\;\;\;;{\delta_{k}} = 2\sqrt {{\alpha_{k}}{\beta_{k}}} \\
							{{\frac{{{\delta_{k}} }}{{2{\sqrt {{\alpha_{k}}{\beta_{k}}} }}}}}\;\;\;\;\;\;\;\;\;\;\;\;\;\;\;\;\;\;\;\;\;;0\leq{\delta_{k}} < 2\sqrt {{\alpha_{k}}{\beta_{k}}}.
						\end{array} \right.
					\end{align*}
	\end{proposition}
	\begin{proof}
	Consider the RFxTCLFs $V_k(\bar x_k,t)=\max {\{ 0,-\mathfrak{H}^k(\bar x_k,t)}\}$ for each predicate $\phi_k$. These functions satisfy $V_k(\bar x_k,t)=0$ for $\bar x_k(0)\in \mathfrak{C}_k(0)$.  Therefore, as long as $\mathfrak{H}^k(\bar x_k,t)\geq 0$, $V_k$ remains $0$ and then $\bar x_k(t)\in \mathfrak{C}_k(t)$, $t\geq 0$.
	Moreover, $V_k(\bar x_k,t)>0$ for $\bar x_k\in P_k(\mathfrak{D})\backslash \mathfrak{C}_k(t)$ and 
	\begin{equation*}
	\dot V_k(\bar x_k,t)\leq {\delta _{k}}- {\alpha _k}	{V_k(\bar x_k,t)}^{\gamma _{1k}} -{\beta _k} {	V_k(\bar x_k,t)}^{\gamma _{2k}}.
	\end{equation*}
	Thus, according to Lemma \ref{RFxT}, the convergence of $V_k(\bar x_k,t)$ to the set $D_k$ in a fixed-time $T_k$ is guaranteed. In other words, 
		\begin{align}\label{D3}
						\mathfrak{H}^k({{\bar x}_k},t) \geq \left\{ \begin{array}{l}
							  -{(\frac{{{\delta_{k}} + \sqrt {{\delta_{k}}^2 - 4{\alpha_{k}}{\beta_{k}}} }}{{2{\alpha_{k}}}})}^{\mu_k}\;\;;{\delta_{k}} > 2\sqrt {{\alpha_{k}}{\beta_{k}}} \\
							-{k_k^{{\mu_k}} }{{(\frac{{{\beta_{k}}}}{{{\alpha_{k}}}})}^{\frac{\mu_k }{2}}}\;\;\;\;\;\;\;\;\;\;\;\;\;\;\;;{\delta_{k}} = 2\sqrt {{\alpha_{k}}{\beta_{k}}} \\
							-{{\frac{{{\delta_{k}} }}{{2{\sqrt {{\alpha_{k}}{\beta_{k}}} }}}}}\;\;\;\;\;\;\;\;\;\;\;\;\;\;\;\;\;\;\;\;\;;0\leq{\delta_{k}} < 2\sqrt {{\alpha_{k}}{\beta_{k}}}\\
							0\;\;\;\;\;\;\;\;\;\;\;\;\;\;\;\;\;\;\;\;\;\;\;\;\;\;\;\;\;\;\;\;\;;{\delta_{k}\leq 0},
						\end{array} \right.
					\end{align}
			which ensures the convergence to set $\mathfrak{C}_{k,rf}(t)$.
\end{proof}

%	As conclusion, we can use fixed-time control barrier function to design controllers which ensure that the system enters a neighborhood of the desired sets within a fixed-time. 
	%and stays in those sets thereafter. 
			\begin{remark}
				Note that in the presence of non-vanishing disturbances, it is not possible to guarantee the convergence of state trajectories to the desired set $\mathfrak{C}_k$. The set $\mathfrak{C}_{k,rf}$ gives an estimate of the neighborhood that the system trajectories converge to, within a fixed-time interval  upper-bounded by \eqref{T}.
				In the cases that the system dynamics does not contain any couplings and task conflictions ($\delta_k=0$), the convergence to $\mathfrak{C}_k$ is guaranteed.  
				As we consider the conflicting specifications and couplings between the agents and model them by constant upper-bounds, the system contains non-vanishing disturbance and hence, Lemma \ref{RFxT} is applied here.  
				\end{remark}
\subsection{QP based formulation}\label{QP_for}
We now formulate a quadratic program that renders $\mathfrak{C}_k(t)$ \emph{robust fixed-time convergent} in the presence of dynamic couplings as well as task conflictions. 
Define $z_k=\left[u_k^T, \varepsilon_k\right]^T\in\mathbb{R}^{m_k+1}$, and consider the following optimization problem to find a control input that solves Problem 1.
%\begin{subequations}\label{QP}
\begin{align}
&\mathop {\min }\limits_{u_k,\varepsilon_k\in\mathbb{R}_{\geq 0}} \frac{1}{2}{z_k^T}z_k\nonumber\\
&\rm{s.t.}\;\;\;
%\it{A_{uk}u_k\leq b_{uk}},\label{u}\\
\it{\frac{{\partial {	\mathfrak{H}^k({{\bar x}_k},t)}}}{{\partial {x_k}}}\left( {{f_k}({x_k},t) + {g_k}({x_k},t){u_k}} \right)
+ \frac{{\partial {	\mathfrak{H}^k({{\bar x}_k},t)}}}{{\partial t}}}\nonumber\\& \;\;\;\;\;\;\;\;\geq\delta_k- {\alpha _k}\; \sgn({	\mathfrak{H}^k({{\bar x}_k},t)}){	|\mathfrak{H}^k({{\bar x}_k},t)|}^{\gamma _{1k}}\nonumber\\& \;\;\;\;\;\;\;\; -{\beta _k}\; \sgn({	\mathfrak{H}^k({{\bar x}_k},t)}){	|\mathfrak{H}^k({{\bar x}_k},t)|}^{\gamma _{2k}}-\varepsilon_k,\label{CLF}
%\\
%&\frac{{\partial {h_s}_{kj}(x_k,x_j)}}{{\partial {x_k}}}\left( {{f_k}({x_k},t) + {g_k}({x_k},t){u_k}} \right) \ge -{\beta _{kj}}{h_{{s_{kj}}}} - 2(x_k-x_j)^T\dot{x}_j, j\in \mathcal{N}_k(t),\label{collision}
%\\
%&\frac{{\partial {\mathfrak{b}^k}({{\bar x}_k},t)}}{{\partial {x_i}}}\left( {{f_i}({x_i},t) + {g_i}({x_i},t){u_i}} \right)
%+\frac{{\partial {\mathfrak{b}^k}({{\bar x}_k},t)}}{{\partial t}}\nonumber\\&\;\;\;\;\;\;\;\;\ge  - {\alpha _k}\left( {{\mathfrak{b}^k}({{\bar x}_k},t)} \right)+\left\| {\frac{{\partial {\mathfrak{b}^k}({{\bar x}_k},t)}}{{\partial {\bar x}_k}}} \right\|\tilde C_k
%\label{CBF}
\end{align}
%\end{subequations}
%where $\delta_k\geq\left\| {\frac{{\partial \mathfrak{H}^k({{\bar x}_k},t)}}{{\partial {\bar x}_k}}} \right\|\tilde C_k$ for all $(\bar x_k,t)\in P_k(\mathfrak{D})\times\mathbb{R}_{\geq 0}$.
%The constraint \eqref{u} encodes the control input constraint. 
where $\delta_k\geq\left\| {\frac{{\partial \mathfrak{H}^k({{\bar x}_k},t)}}{{\partial {\bar x}_k}}} \right\|\tilde C_k$. Constraint \eqref{CLF} corresponds to the fixed-time convergence of the closed-loop trajectories to the set $\mathfrak{C}_{k,rf}(t)$, where ${\alpha _{k}}, {\beta _{k}}>0$, $\gamma_{1k}=1+\frac{1}{\mu_k}$, $\gamma_{2k}=1-\frac{1}{\mu_k}$, $\mu_k>1$ are fixed.  Moreover, $\varepsilon_k\geq0$ relaxes QP in the presence of conflicting tasks and minimizing it results in a least violating solution.
%    	\mathfrak{C}_{k,\max}(t):=\{\bar x_k\in\mathbb{R}^{\bar n_k}|\mathfrak{H}^k(\bar x_k,t)\geq \alpha_k^{-1}(-\epsilon_{k,\max})\},
%\end{align*}
%where $\epsilon_{k,\max}=\sup_{(\bar x_k,t)\in P_k(\mathfrak{D}_k)\times\mathbb{R}_{\ge 0}}\hat\epsilon_k(\bar x_k,t)$ and $\alpha_k:\mathbb{R}\to\mathbb{R}$ is an extended class $\mathcal{K}$ function, see~\cite[Corollary 2]{lindemann2019control}.
%The inter-agent collision avoidance between agent $k\in\mathcal{V}$ and its neighbors $j\in\mathcal{N}_k(t)$ is stated in \eqref{collision},
%where $h_{{s_{kj}}}(x_k,x_j)=\left\| {x_k-x_j} \right\|^2-d_s^2\geq 0$, $\beta _{kj}\in\mathbb{R}_{\geq 0}$.
	\begin{remark}\label{re1}
	Note that our analysis relies on Assumption \ref{ass2}. However, this assumption is obsolete if \eqref{CLF} is solved for each agent $k$. Thus, to give an estimation on $\tilde C_k$, first the set $\mathfrak{D}$ should be selected such that $P_k(\mathfrak{D}) \supset \mathfrak{B}_k$ for each $k$. Then, $\tilde C_k$ is selected such that $\left\| {\tilde c_k(x,t)} \right\|\leq \tilde C_k$ for all $(x,t)\in\mathfrak{D}\times\mathbb{R}_{\geq 0}$. Assuming that the agents are subject to bounded inputs, i.e., $u_k(t)\in\mathcal{U}_k$ for some compact set $\mathcal{U}_k$, an estimate of $\tilde C_k$ can be obtained.
	In addition, considering Assumption \ref{concave} and barrier function construction according to Section \ref{solution3}, $\left\| {\frac{{\partial \mathfrak{H}^k({{\bar x}_k},t)}}{{\partial {\bar x}_k}}} \right\|$ is upper bounded and this bound can be acquired, too.
	\end{remark}	
\begin{theorem}\label{feasiblity}
Let the solution to the QP \eqref{CLF} be denoted as $z_k^*(\cdot)$. Assume that $\delta_k\geq\tilde C_k\left\| {\frac{{\partial \mathfrak{H}^k({{\bar x}_k},t)}}{{\partial {\bar x}_k}}} \right\|-\varepsilon_k^*$  for all $(\bar x_k,t)\in P_k(\mathfrak{D})\times\mathbb{R}_{\geq 0}$. If the solution $z_k^*(\cdot)$ is continuous on $P_k(\mathfrak{D})\backslash \mathfrak{C}_k(t)$, then under the control input $u_k(\cdot)=u_k^*(\cdot)$ the closed-loop trajectories of \eqref{multi_sys}  reach the set $\mathfrak{C}_{k,rf}$ in a fixed-time $T_{k}$ given by \eqref{T}, with $a_{1k}=\alpha_k$, $a_{2k}=\beta_k$ and $a_{3k}=\delta_k$.
\end{theorem}
\begin{proof}
Considering Proposition \ref{convergence} for $\delta_k\geq\left\| {\frac{{\partial \mathfrak{H}^k({{\bar x}_k},t)}}{{\partial {\bar x}_k}}} \right\|\tilde C_k-\varepsilon_k$, convergence to the set $\mathfrak{C}_{k,rf}$, which is provided by the user-defined bounds for $\mathfrak{H}^k({{\bar x}_k},t)$ as in \eqref{D3}, will be achieved in the presence of couplings and task conflictions in a least violating way. 
\end{proof}
\section{Simulations}\label{sim}
Consider a multi-agent system consisting of $M:=3$ omnidirectional robots denoting by $x_k:=\left[ {p_k^T ,x_{k,3}} \right]^T\in\mathbb{R}^3$, $k\in\{{1,\cdots,M}\}$, in which $p_k:=\left[ {x_{k,1} ,x_{k,2}} \right]^T$ and  $x_{k,3}$ represent the robot's position and orientation with respect to the first coordinate, respectively  \cite{liu2008omni}.
The agent dynamics are subject to $\dot x_k=f_k(x,t)+g_k(x_k,t)u_k+c_k(x,t)$, where $g_k:=\left[ {\begin{array}{*{20}{c}}
{\cos ({x_{k,3}})}&{ - \sin ({x_{k,3}})}&0\\
{\sin ({x_{k,3}})}&{\cos ({x_{k,3}})}&0\\
0&0&1
\end{array}} \right](B_k^T)^{-1}R_k$ with $B_k:=\left[ {\begin{array}{*{20}{c}}
0&{\cos (\pi/6)}&{-\cos (\pi/6)}\\
-1&{\sin (\pi/6)}&{\sin (\pi/6)}\\
L_k&L_k&L_k
\end{array}} \right]$ to model the geometric constraint. Moreover, $R_k=0.02$ is the wheel radius and $L_k=0.2$ describes the radius of the robot body.
Furthermore, $f_k(x,t)$ are locally Lipschitz continuous functions describing the induced dynamical couplings for the purpose of collision avoidance.
We follow the the example of \cite{lindemann2018control} by adding the coupling effects of other agents and considering different tasks to show the effect of changing the parameters during the conflictions.
We pick $\tilde C_k=1$, $k\in\{{1,2,3}\}$, to model the disturbances or conflicting behavior of other agents. 
Consider the formulae
  $\phi_1:=G_{\left[ {15,90} \right]}(\left\| {{p_1} + {g_1} - {p_2}} \right\| \le {2})\wedge G_{\left[ {25,35} \right]}(\left\| {{p_1} + {g_2} - {p_3}} \right\| \le {7.7})\wedge F_{\left[ {50,90} \right]}(\left\| {{p_1} - {g_3}} \right\| \le {2})$, $\phi_2:=G_{\left[ {15,90} \right]}(\left\| {{p_2} - {g_1} - {p_1}} \right\| \le {2})\wedge F_{\left[ {30,35} \right]}(\left\| {{p_2}- {g_4}} \right\| \le {4})\wedge F_{\left[ {50,90} \right]}(\left\| {{p_2} +{g_1}- {p_3}} \right\| \le {5})$,
  $\phi_3:=G_{\left[ {25,35} \right]}(\left\| {{p_3} - {g_2} - {p_1}} \right\| \le {7.7})
  \wedge F_{\left[ {40,60} \right]}(\left\| {{p_3} - {g_5}} \right\| \le {5})\wedge F_{\left[ {50,90} \right]}(\left\| {{p_3} -{g_1}- {p_2}} \right\| \le {5})$, 
where $g_1=\left[ {0.8,0} \right]^T$, $g_2=\left[ {0,-0.8} \right]^T$, $g_3=\left[ {-1.2,1.2} \right]^T$, $g_4=\left[ {1.2,1.2} \right]^T$, $g_5=\left[ {1.2,-1.2} \right]^T$.
We first choose the parameters of the QP formulation as $\mu_k=4$, $\alpha_k=\beta_k=1$, $k\in\{{1,2,3}\}$. Then, we get $\delta_k=0.9741$ and considering \eqref{D3}, the upper bound of $-{{\frac{{{\delta_{k}} }}{{2{\sqrt {{\alpha_{k}}{\beta_{k}}} }}}}}=-0.487$ is acquired for the TFCBFs $\mathfrak{H}^k(\bar x_k,t)$, $k\in\{{1,2,3}\}$, as can be seen in Figure \ref{barrier} as well as the agent trajectories presented in Figure \ref{traj}.
We change the value of parameters $\alpha_k, \beta_k$ to $0.4$. This leads to $\delta_k=0.9713$ and $\mathfrak{H}^k(\bar x_k,t)\geq -{(\frac{{{\delta_{k}} + \sqrt {{\delta_{k}}^2 - 4{\alpha_{k}}{\beta_{k}}} }}{{2{\alpha_{k}}}})}^{\mu_k}=13.09$, $k\in\{{1,2,3}\}$. Therefore, a higher deviation from the desired set as well as a faster settling-time in the main switching instant of $t=35 \;\rm{s}$ is acquired as is shown in Figure \ref{barrier2}. 
Hence, the fixed-time convergence criterion allows us to characterize the behavior of TFCBFs independent of the agents initial conditions.
The computation times on an Intel Core i5-8365U with 16 GB of RAM are about $2.45$ms.
	\begin{figure}
		\centering
		\hspace*{-0.5cm}
		\psfragscanon 
			\psfrag{a1111111111}[Bc][B1][0.35][0]{$\mathfrak{H}^1(\bar x_1,t)$}
			\psfrag{b1111111111}[Bc][B1][0.35][0]{$\mathfrak{H}^2(\bar x_2,t)$}
			\psfrag{c1111111111}[Bc][B1][0.35][0]{$\mathfrak{H}^3(\bar x_3,t)$}
		\includegraphics[width=6cm]{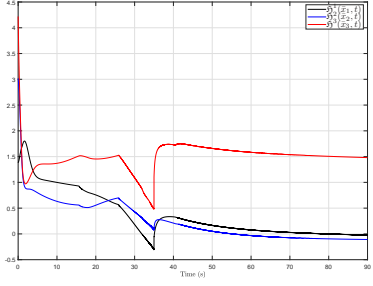}
		\caption{Fixed-time convergent barrier functions evolution for $\alpha_k=\beta_k=1$.}
		\label{barrier}
		 \psfragscanoff
	\end{figure}
		\begin{figure}
		\centering
		\hspace*{-0.5cm}
		\includegraphics[width=6cm]{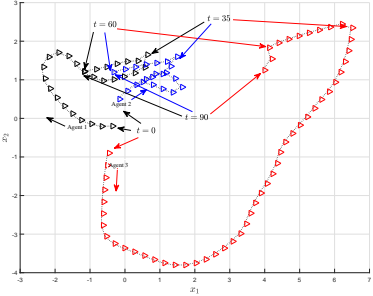}
		\caption{Robot trajectories. The triangles denote the orientation for $\alpha_k=\beta_k=1$.}
		\label{traj}
	\end{figure}
		\begin{figure}
		\centering
		\hspace*{-0.5cm}
		\psfragscanon 
			\psfrag{d1111111111}[Bc][B1][0.35][0]{$\mathfrak{H}^1(\bar x_1,t)$}
			\psfrag{e1111111111}[Bc][B1][0.35][0]{$\mathfrak{H}^2(\bar x_2,t)$}
			\psfrag{f1111111111}[Bc][B1][0.35][0]{$\mathfrak{H}^3(\bar x_3,t)$}
		\includegraphics[width=6cm]{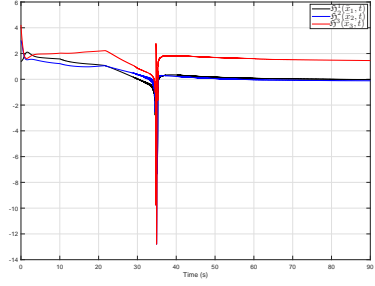}
		\caption{Fixed-time convergent barrier functions evolution for $\alpha_k=\beta_k=0.4$.}
		\label{barrier2}
		 \psfragscanoff
	\end{figure}

\section{Conclusion}\label{conc}
Based on a new notion of time-varying fixed-time convergent control barrier functions, we presented a feedback control strategy to find robust solutions for the performance of the multi-agent systems under conflicting local STL tasks. In particular, the lower bound of the introduced TFCBFs and the finite convergence time can be characterized in a user-specified way, independent of the initial conditions of the agents. Future works extend these results to more general collaborative tasks, including leader-follower topologies.

\bibliographystyle{IEEEtran}
\bibliography{references_STL}

%>>>>>>> 5f968ccae93640ed0e9c350e78b160d149579522
\end{document}